\documentclass[letterpaper, 10 pt, conference]{ieeeconf}  
\IEEEoverridecommandlockouts  
\usepackage[utf8]{inputenc}
\usepackage{xcolor}
\usepackage{graphicx}
\usepackage{graphics}
\usepackage{subcaption}
\usepackage{amsmath}
\usepackage[version=4]{mhchem}
\usepackage{siunitx}
\usepackage{longtable,tabularx}
\usepackage{algorithm}
\usepackage{algorithmicx}
\usepackage{algpseudocode}
\usepackage{amsmath}
\usepackage{amssymb}
\usepackage{proof}
\usepackage{epsfig}
\usepackage{calc}

\usepackage{amsthm}
\usepackage{accents}
\usepackage{cite}

\newtheorem{remark}{\textnormal{\textbf{Remark}}}
\newtheorem{theorem}{\textnormal{\textbf{Theorem}}}

\usepackage[autostyle]{csquotes}

\newcommand{\bx}{{\bf x}}
\newcommand{\by}{{\bf y}}
\newcommand{\bz}{{\bf z}}

\newcommand{\bv}{{\bf v}}
\newcommand{\bw}{{\bf w}}

\title{\LARGE \bf
Optimal Sensor Gain Control for Minimum-Information Estimation of Continuous-Time Gauss-Markov Processes}

\author{Vrushabh Zinage$^{1}$ \and Takashi Tanaka$^{2}$ \and Valeri Ugrinovskii$^{3}$ 
\thanks{$^{1}$Vrushabh Zinage is a graduate student at the Department of Aerospace Engineering and Engineering Mechanics, University of Texas at Austin, Austin, Texas, 78712 
         {\tt\small vrushabh.zinage@utexas.edu }}%
 \thanks{$^{2}$Takashi Tanaka is an Assistant Professor at the Department of Aerospace Engineering and Engineering Mechanics, University of Texas at Austin, Austin, Texas, 78712       {\tt\small ttanaka@austin.utexas.edu}}
 \thanks{$^{3}$Valeri Ugrinovskii is a Professor at the School of Engineering and
 IT, The University of New South Wales, Canberra, Australia.{\tt\small v.ougrinovski@adfa.edu.au }}
 \thanks{V. Zinage and T. Tanaka were supported by NSF Award 1944318. V. Ugrinovskii was supported by the Australian Research Council under 
 Discovery Projects funding scheme (project DP200102945). }
  }

\begin{document}

\bibliographystyle{IEEEtran} 

\maketitle
\thispagestyle{empty}
\pagestyle{empty}

\begin{abstract}
We consider the scenario in which a continuous-time  Gauss-Markov process is estimated by the Kalman-Bucy filter over a Gaussian channel (sensor) with a variable sensor gain. The problem of scheduling the sensor gain over a finite time interval to minimize the weighted sum of the data rate (the mutual information between the sensor output and the underlying Gauss-Markov process) and the distortion (the mean-square estimation error) is formulated as an optimal control problem. A necessary optimality condition for a scheduled sensor gain is derived based on Pontryagin's minimum principle. 
For a scalar problem, we show that an optimal sensor gain control is of bang-bang type, except the possibility of taking an intermediate value when there exists a stationary point on the switching surface in the phase space of canonical dynamics. Furthermore, we show that the number of switches is at most two and the time instants at which the optimal gain must be switched can be computed from the analytical solutions to the canonical equations.
\end{abstract}

\section{Introduction}
In this paper, we consider a controlled sensing problem in which a continuous-time linear-Gaussian random process is observed through a linear sensor whose sensor gain is strategically adjusted. The sensor gain is optimized to minimize the weighted sum of the minimum mean-square error (MMSE) attainable by a causal estimator and the mutual information (I) between the underlying random process and the sensor output. Generally, the former cost is reduced by adopting a large sensor gain, while this leads to an increased cost in the latter sense. Therefore, these two performance criteria are in a trade-off relationship (I-MMSE relationship \cite{duncan1970calculation,kadota1971mutual,guo2005mutual,palomar2005gradient}), and attaining a sweet spot by an optimal sensor gain is a nontrivial problem.

The problem studied in this paper is motivated by a practical scenario where a continuous-time source signal is encoded, compressed, and transmitted to a remote user where the signal is reproduced in a zero-delay manner (e.g., the event-based camera \cite{gallego2019event} for robotics applications). Assuming that binary codewords are used for communication, the trade-off between the bit-rate and the best attainable quality of the reproduced signal is of our natural interest. 
Although designing the optimal architecture (e.g., the optimal spatio-temporal sampling and encoding schemes) is a challenging task, it can be shown \cite{tanaka2017optimal} that the minimum bit-rate required for reproducing the signal within a given distortion criterion is lower bounded by the aforementioned mutual information required for reproducing the signal within the same criterion. 
Thus, the solution we present in this paper reveals a fundamental performance limitation of the sensor-encoder joint design for such a remote estimation problem. 

In the literature, related problems are studied in the context of sequential rate-distortion problems (the rate-distortion problems with causality constraints)
\cite{gorbunov1974prognostic, tatikonda2004stochastic,derpich2012improved,charalambous2013nonanticipative,tanaka2016semidefinite} and their applications to control problems \cite{kostina2019rate,tanaka2017lqg}.
Most of the existing works consider discrete-time source signals. 
For discrete-time Gauss-Markov sources and mean-square distortion criteria, \cite{tatikonda2004stochastic} shows that the optimal policy (test channel) to the sequential rate-distortion problem is linear. 
Based on this observation, \cite{tanaka2016semidefinite} shows that the Gaussian sequential rate-distortion problem is equivalent to an optimal sensor gain control problem, which was shown to be solvable by means of semidefinite programming.
A continuous-time counterpart of the same problem is considered in \cite{tanaka2017optimal}, although only infinite-horizon, time-invariant cases are studied there. 

\textit{Main contributions:} In this paper, we expand the problem considered in \cite{tanaka2017optimal} to finite-horizon, time-varying cases.
We first show that the optimal sensor gain control problem can be formulated as a nonlinear optimal control problem to which Pontryagin's minimum principle is applicable. 
To gain further insight into the optimal solution, we then restrict ourselves to scalar (single-dimensional) problems where a feasible control gain (control input) is constrained to be in $[0,1]$. 
We prove that the optimal sensor gain is of bang-bang type characterized by a switching surface in the phase space, except the possibility of taking an intermediate value if the canonical equation admits an equilibrium on the switching surface. {Further, we provide a method to compute the optimal sensor gain from the analytical expressions to the canonical equations. Consequently, the optimal sensor gain for the original sensor selection problem can be directly computed from the optimal control input. }

  The paper is organized as follows. Section \ref{sec:problem_statement} introduces the problem statement followed by preliminaries in Section \ref{sec:preliminaries}. Section \ref{sec:main_result} discusses the main result of the paper followed by some concluding remarks and future work in Section \ref{sec:conclusion}.

Notation: We will use notation $x_{[t_1, t_2]}=\{x_t: t_1 \leq t \leq t_2\}$ to denote a continuous-time signal. Bold symbols like $\bx$ will be used to denote random variables. We assume all the random variables considered in this paper are defined on the probability space $(\Omega, \mathcal{F}, \mathcal{P})$. The probability distribution $\mu_\bx$ of an $(\mathcal{X}, \mathcal{A})$-valued random variable $\bx$ is defined by
\[
\mu_\bx(A)=\mathcal{P}\{\omega\in\Omega : \bx(\omega)\in A\}, \; \forall A \in \mathcal{A}.
\]
If $\bx_1$ and $\bx_2$ are both $(\mathcal{X}, \mathcal{A})$-valued random variables, the relative entropy from $\bx_2$ to $\bx_1$ is defined by
\[
D(\mu_{\bx_1} \| \mu_{\bx_2})=\int \log \frac{d\mu_{\bx_1}}{d\mu_{\bx_2}} d\mu_{\bx_1}
\]
provided that the Radon-Nikodym derivative $\frac{d\mu_{\bx_1}}{d\mu_{\bx_2}}$ exists, and $+\infty$ otherwise. The mutual information between two random variables $\bx$ and $\by$ is defined by
\[
I(\bx;\by)=D(\mu_{\bx\by} \| \mu_\bx \otimes \mu_\by)
\]
where $\mu_{\bx\by}$ and $\mu_\bx \otimes \mu_\by$ denote the joint and product distributions, respectively.

\section{Optimal Sensor Design and Problem Statement\label{sec:problem_statement}}
Similarly to \cite{tanaka2017optimal}, we consider the scenario of estimating an $n$-dimensional Ito process based on an $n$-dimensional noisy measurement.
Let $(\Omega,\mathcal{F}, \mathcal{P})$ be a complete probability space and suppose $\mathcal{F}_t\subset \mathcal{F}$ be a non-decreasing family of $\sigma$-algebras.
Let $(\bw_t, \mathcal{F}_t)$ and $(\bv_t, \mathcal{F}_t)$ be $\mathbb{R}^n$-valued, mutually independent standard Wiener processes with respect to $\mathcal{P}$.
The random process to be estimated is an $n$-dimensional Gauss-Markov process defined by
\begin{equation}
\label{eq:source}
    d\bx_t=A\bx_t dt + B d\bw_t, \;\;\; t\in [t_0, t_1]
\end{equation}
with $\bx_{t_0}\sim\mathcal{N}(0, X_0)$, where $X_0\succeq 0$ is a given covariance matrix and $t_0<t_1<\infty$ is the horizon length of the considered problem. We assume $(A,B)$ is a controllable pair of matrices. Let $C_t:[t_0, t_1] \rightarrow \mathbb{R}^{n\times n}$ be a measurable function representing the time-varying sensor gain.
Setting $\bz_t \triangleq C_t \bx_t$, the sensor mechanism produces an $n$-dimensional signal
\begin{equation}
\label{eq:measurement}
    d\by_t=\bz_t dt + d\bv_t, \;\;\; t\in [t_0, t_1]
\end{equation}
with $\by_{t_0}=0$. Based on the sensor output $\by_t$, the causal MMSE estimate $\hat{\bx}_t\triangleq \mathbb{E}(\bx_t|\mathcal{F}_t^{\by})$ is computed, where $\mathcal{F}_t^{\by} \subset \mathcal{F}$ denotes the  $\sigma$-algebra generated by $\by_{[0,t]}$. Computationally, this can be achieved by the 
Kalman-Bucy filter
\begin{equation}
\label{eqkbfilter}
d\hat{\bx}_t=A\hat{\bx}_tdt+X_tC_t^\top (d\by_t-C_t\hat{\bx}_tdt), \;\;\; t\in [t_0, t_1]
\end{equation}
with $\hat{\bx}_{t_0}=0$. Here, $X_t$ is the unique solution to the matrix Riccati differential equation
\begin{equation}
\label{eqriccati}
\dot{X_t}=AX_t+X_tA^\top-X_tC_t^\top C_tX_t+BB^\top, \;\; t\in [t_0, t_1]
\end{equation}
with the aforementioned given initial covariance matrix $X_{t_0}=X_0\succeq 0$. 
The overall architecture of the sensor mechanism is shown in Fig.~\ref{fig:sensor}.
\begin{figure}[h]
\centering
\includegraphics[width=8.5cm]{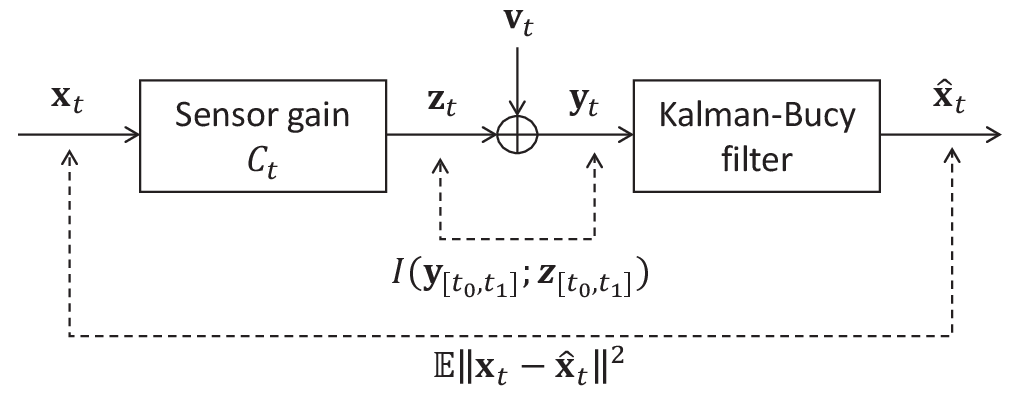}
\caption{Architecture of the sensor mechanism and the performance criteria.}
\label{fig:sensor}
\end{figure}

\subsection{Performance Criteria}
Unlike the standard filtering problem where the sensor gain $C_t$ is given, in this paper, we consider the problem of optimally choosing $C_t$ to achieve the minimum-information filtering. The optimality is characterized in terms of two performance measures, namely the mean-square error and the mutual information, as defined below.

\subsubsection{Mean-square error (MSE)}
The first criterion is the MMSE over the considered time horizon.
\begin{equation}
\label{eq:mmse}
 \int_{t_0}^{t_1} \mathbb{E} \| \bx_t-\hat{\bx}_t\|^2 dt =\int_{t_0}^{t_1} \mathrm{Tr} (X_t) \; dt.   
\end{equation}

\subsubsection{Mutual information}
The second performance criterion is the mutual information $I(\by_{[t_0, t_1]};\bz_{[t_0, t_1]} )$.
We use the following key result due to Duncan \cite{duncan1970calculation}:
\begin{theorem}
\label{theo:MI}
\normalfont Let the random processes $\by_{[t_0,t_1]}$ and $\bz_{[t_0,t_1]}$ be defined as above. Then
\[
I(\by_{[t_0,t_1]}; \bz_{[t_0,t_1]})=\frac{1}{2}\int_{t_0}^{t_1} \mathbb{E}\|C_t (\bx_t-\hat{\bx}_t)\|^2 dt.
\]
\end{theorem}
Theorem~\ref{theo:MI} together with equation \eqref{eq:mmse} implies that the mutual information $I(\by_{[t_0,t_1]};\bz_{[t_0,t_1]})$ can also be expressed as 
\begin{equation}
    I(\by_{[t_0,t_1]};\bz_{[t_0,t_1]})=\frac{1}{2}\int_{t_0}^{t_1}\mathrm{Tr}(C_tX_tC_t^\top)dt. \label{eqn:mutual_information}
\end{equation}

\subsection{Problem Statement}
Minimizing the MSE \eqref{eq:mmse} and the mutual information \eqref{eqn:mutual_information} are in a trade-off relationship.\footnote{Suppose we choose $C_t=kC \;\; \forall t$ where $k\geq 0$ is a scalar and $(A,C)$ is an observable pair. As $k\rightarrow +\infty$, the MSE tends to zero while the mutual information tends to $+\infty$.} Thus, our task is to choose a measurable function $C_t:[t_0, t_1]\rightarrow \mathbb{R}^{n\times n}$ to minimize the weighted sum of them. Introducing a trade-off parameter $\alpha>0$,\footnote{The parameter $\alpha$ is the Lagrange multiplier in view of the hard-constrained version of the problem studied in \cite{tanaka2017optimal}.} the main problem we study in this paper is formulated as follows:
\begin{subequations}
\label{eqn:main_optimization_problem}
\begin{align}
   \underset{C_t}{\mathrm{min}}&\;\; \int_{t_0}^{t_1}\;\mathbb{E}\|\bx_t-\hat{\bx}_t\|^2dt+2\alpha I(\by_{[t_0,t_1]};\bz_{[t_0,t_1]}) \label{eqn:main_optimization_problem1} \\
   \text{ s.t. }&\;\;\; \mathrm{Tr}(C_t^\top C_t) \leq 1 \;\; \forall \;\;t\in [t_0, t_1]. \label{eqn:main_optimization_problem2}
\end{align}
\end{subequations}
The constraint \eqref{eqn:main_optimization_problem2} is introduced to incorporate the fact that the allowable sensor gain usually has an upper bound. Using \eqref{eq:mmse} and \eqref{eqn:mutual_information}, the problem \eqref{eqn:main_optimization_problem} can be written as
\begin{subequations}
\begin{align}
     \underset{C_t}{\mathrm{min}}\quad&\int_{t_0}^{t_1}\mathrm{Tr}(X_t)dt+\alpha \int_{t_0}^{t_1}\mathrm{Tr}(C_tX_tC_t^\top)dt\\
    \mathrm{s.t.}\quad&\dot{X}_t=AX_t+X_tA^\top-X_tC_t^\top C_tX_t+BB^\top \\
    &X_{t_0}=X_{0}\\
    &\mathrm{Tr}(C_t^\top C_t)\leq 1\quad\forall\;\; t\in[t_0,t_1].
\end{align}
\label{problem:optimization_problem}
\end{subequations}
Introducing $U_t\triangleq C_t^\top C_t\succeq 0$, this can be written as an equivalent optimal control problem with state $X_t$ and control input $U_t$:
\begin{subequations}
\begin{align}
    \underset{U_t}{\mathrm{min}}\quad&\int_{t_0}^{t_1}\mathrm{Tr}(X_t+\alpha U_tX_t)dt  \label{problem:optimization_problem_control_a}\\
    \mathrm{s.t.}\quad&\dot{X}_t=AX_t+X_tA^\top-X_tU_tX_t+BB^\top\\
    &X_{t_0}=X_{0}\\
    &U_t\geq0,\;\;\;\mathrm{Tr}(U_t)\leq 1\quad\forall\;\; t\in[t_0,t_1].\label{eqn:admissible_controls}
\end{align}
 \label{problem:optimization_problem_control}
\end{subequations}
The minimization is over the space of measurable functions $U_t: [t_0, t_1]\rightarrow \mathbb{S}_+^n (=\{M\in \mathbb{R}^{n\times n}: M \succeq 0\})$.
\begin{remark}
\normalfont The equivalence between
\eqref{problem:optimization_problem} and \eqref{problem:optimization_problem_control} implies that optimal solutions to the main problem \eqref{eqn:main_optimization_problem}, if exist, are not unique. Namely, if $U_t^*$ is an optimal solution to  \eqref{problem:optimization_problem_control}, then both $\bar{C}_t$ and $\tilde{C}_t$ are optimal solutions to \eqref{problem:optimization_problem} if $U_t^*=\bar{C}_t^\top \bar{C}_t = \tilde{C}_t^\top \tilde{C}_t$.
\end{remark}

\begin{remark}
\normalfont We assume that the dimension of the sensor output $\by_t$ matches the dimension of the underlying Gauss-Markov process $\bx_t$. 
This assumption is needed to avoid a technical difficulty that arises when $\by_t$ needs to be $m$-dimensional and $m<n$. To see this, suppose that the sensor gain $C_t$ to be synthesized in \eqref{eqn:main_optimization_problem} is required to be $\mathbb{R}^{m\times n}$-valued and $m<n$. Since this implies 
$\text{rank}(C_t^\top C_t)\leq m$, an additional non-convex constraint $\text{rank}(U_t)\leq m \; \forall t\in[t_0, t_1]$ must be included in \eqref{problem:optimization_problem_control} to maintain the equivalence between \eqref{problem:optimization_problem} and 
\eqref{problem:optimization_problem_control}. This type of difficulty has been observed in the sensor design problems in the literature, as in \cite{bansal1989simultaneous}. 
\end{remark}

\section{Optimality condition\label{sec:preliminaries}}
\subsection{Existence of optimal control}
{
We first remark that the following result due to Filippov \cite{filippov1962certain} is applicable to guarantee the existence of an optimal control for solving problem \eqref{problem:optimization_problem_control}:
 

\begin{theorem}
\label{theo:filippov}
\normalfont \textbf{(Filippov's theorem): \cite{optimal_control_daniel_liberzon}}
Given a control system $\dot{x}_t=f(x_t,u_t)$ with $u_t\in\mathcal{U}$, assume that its solutions exist on a time interval $[t_0,t_1]$ for all controls and that for every $x_t$ the set $\{f(x_t,u_t):u_t\in\mathcal{U}\}$ is compact and convex. Then the reachable set $R_t(x_0)$ is compact for each $t\in[t_0,t_1]$.
\end{theorem}
Specifically, one can convert the original Langrange-type problem \eqref{problem:optimization_problem_control} into an equivalent Mayer-type problem by introducing an auxiliary state $x_t^{\text{aux}}$ satisfying $x_{t_0}^{\text{aux}}=0$ and $\dot{x}_t^{\text{aux}}=\text{Tr}(X_t+\alpha U_t X_t)$. 
In the Mayer form, the original problem of minimizing \eqref{problem:optimization_problem_control_a} becomes the problem of minimizing $x_{t_1}^{\text{aux}}$ over the reachable set at $t=t_1$. Since the premises of Theorem~\ref{theo:filippov} are satisfied by the obtained Mayer-type problem, we can conclude that the reachable set at $t=t_1$ is compact. Therefore, Weierstrass' extreme value theorem guarantees the existence of an optimal solution. 
}
\subsection{Pontryagin Minimum Principle}
In this section, we briefly discuss the minimum principle for fixed-endtime free-endpoint optimal control problem followed by its application to our problem \eqref{problem:optimization_problem_control}.
Suppose that an admissible control input is a measurable function $u_t:[t_0,t_1]\rightarrow \mathcal{U}$ where $\mathcal{U}\subset \mathbb{R}^m$ is a compact set. We assume that $f(x_t,u_t)$, $\frac{\partial f(x_t,u_t)}{\partial x_t}$, $L(x_t,u_t)$ and $\frac{\partial L(x_t,u_t)}{\partial x_t}$ are continuous on $\mathbb{R}^n\times \mathcal{U}\times (t_0,t_1)$. Consider the Lagrange-type problem 
\begin{subequations}
\begin{align}
    &\underset{u_t}{\mathrm{min}}\quad \int_{t_0}^{t_1}L(x_\tau,u_\tau)d\tau\\
    &\mathrm{s.t.}\quad \dot{x}_t=f(x_t,u_t),\;\;x_t\in\mathbb{R}^n,\;u_t\in\mathcal{U}\label{eqn:dynamics}\\
    &\quad\;\;\;\;\;x_{t_0}=x_0.
    \label{problem:optimal_control_problem_initial_condition}
\end{align}
\label{problem:optimal_control_problem}
\end{subequations}
\begin{theorem}
\label{theo:pmp}
\normalfont \cite[Theorem 5.10]{athans2013optimal_book}
Suppose there exists an optimal solution to \eqref{problem:optimal_control_problem}. Let $u^\star_t:[t_0,t_1]\rightarrow \mathcal{U}$ be the optimal control input and $x^\star_t:[t_0,t_1]\rightarrow \mathbb{R}^n $ be the corresponding state trajectory. Then, it is necessary that there exists a function $p^\star_t:[t_0,t_1]\rightarrow \mathbb{R}^n$ such that the following conditions hold for the Hamiltonian $H$ defined as
\begin{align}
    H(x_t,p_t,u_t)=L(x_t,u_t)+p^\top_t f(x_t,u_t):
\end{align}

(i) $x^\star_t$ and $p^\star_t$ satisfy the following canonical equations:
\begin{align}
    &\dot{x}^\star_t=\frac{\partial H(x^\star_t,p^\star_t,u^\star_t)}{\partial p_t},\quad\dot{p}^\star_t=-\frac{\partial H(x^\star_t,p^\star_t,u^\star_t)}{\partial x_t}\nonumber
\end{align}
with boundary conditions $x_{t_0}=x_0$ and $p_{t_1}=0$.

(ii) $\underset{u_t\in \mathcal{U}}{\mathrm{min}}\;\; H(x^\star_t,p^\star_t,u_t)=H(x^\star_t,p^\star_t,u^\star_t)$ for all $t\in[t_0,t_1]$.
\end{theorem}

For our problem \eqref{problem:optimization_problem_control}, the Hamiltonian is defined as 
\begin{align}
    H(X_t,P_t,U_t)=&\mathrm{Tr}(X_t+\alpha U_t X_t)+\nonumber\\
    &\left<P_t, AX_t+X_t A^\top-X_tU_tX_t+BB^\top \right>\nonumber\\
    =&\mathrm{Tr}(P_t(AX_t+X_t A^\top+BB^\top))+\mathrm{Tr}(X_t)\nonumber\\
    &+\mathrm{Tr}((\alpha X_t-X_tP_tX_t)U_t).\nonumber
\end{align}
Thus, the necessary optimality condition provided by Theorem~\ref{theo:pmp} is given by the canonical equations
\begin{subequations}
\begin{align}
    &\dot{X}_t=AX_t+X_tA^\top-X_tU_tX_t+BB^\top \\
    &\dot{P}_t=P_tX_tU_t+U_tX_tP_t\!-\!P_tA\!-\!A^\top P_t\!-\!I\!- \!\alpha U_t
\end{align}
\label{eqn:canonical_vector_case}
\end{subequations}
with boundary conditions $X_{t_0}=X_0$ and $P_{t_1}=0$, and
\begin{align}
    U_t^\star=\underset{U_t\in\mathcal{U}}{\mathrm{argmin}}\;\mathrm{Tr}[(\alpha X_t-X_tP_tX_t)U_t]
    \label{eqn:u_star_matrix}
\end{align}
where $\mathcal{U}=\{U_t\in\mathbb{S}_+^n\;\forall t\in[t_0,t_1]: \mathrm{Tr}(U_t)\leq 1\}$.

\section{Optimal solution in scalar case\label{sec:main_result}}

To solve the optimality condition \eqref{eqn:canonical_vector_case}
and \eqref{eqn:u_star_matrix} explicitly, in this section, we restrict our attention to scalar cases (i.e., $n=1$). 
In what follows, we assume $A=a<0$ and $B=1$. The assumption that $a<0$ is natural as it is required for the source process \eqref{eq:source} to be stable.
The canonical equations \eqref{eqn:canonical_vector_case} are simplified as
\begin{subequations}
\label{eqn:region_2_governing_equations_xp}
\begin{align}
    &\dot{x}_t=2ax_t-x_t^2u_t+1\label{eqn:region_2_governing_equations_x}\\
    & \dot{p}_t=2x_tp_tu_t-2ap_t-1-\alpha u_t\label{eqn:region_2_governing_equations_p}
\end{align}
\label{eqn:canonical}
\end{subequations}
with $x_{t_0}=x_0$ and $p_{t_1}=0$. Due to the original meaning of $x_0$ as a covariance of the initial value of the underlying process $\mathbf{x}_0$, we assume that $x_0\geq 0$. Since \eqref{eqn:region_2_governing_equations_x} is a monotone system \cite{angeli2003monotone}, $x_t\geq 0$ for $t\in[t_0,t_1]$. The optimal control $u_t$ is given as follows:
\begin{align}
    u_t^\star&=\underset{0\leq u_t\leq 1}{\mathrm{argmin}}\;\; x_t(\alpha-x_tp_t)u_t\nonumber\\
    &= \left\{ \begin{array}{l}
0\quad\text{if}\quad p_tx_t<\alpha
\\
1\quad\text{if}\quad p_tx_t>\alpha
\\
u^\star\in[0,1]\quad\text{if}\quad p_tx_t=\alpha
\end{array}
\right.
\label{eqn:u_star_argmin}
\end{align}
The main result of this paper is stated as follows:
\begin{theorem}
\normalfont For any $x_0>0$, $\alpha>0$, $a<0$ and specified time interval $[t_0,t_1]$,{ an optimal control exists and satisfies \eqref{eqn:region_2_governing_equations_xp} and \eqref{eqn:u_star_argmin}.
Furthermore, the optimal control is bang-bang type (i.e., it takes either $u_t=0$ or $u_t=1$), except that in a particular case, it can take an intermediate value ($0\leq u^\star\leq 1$ specified by \eqref{eq:u_star}). In all cases, the optimal control is a piecewise constant function with at most two discontinuities.}
\label{theorem:main_theorem}
\end{theorem}
\begin{proof}
{
The existence of an optimal solution and the necessary optimality conditions were discussed in Section~\ref{sec:preliminaries}.
The rest of the statement follows from the analysis in Sections~\ref{subsec:phase_potrait} and Section~\ref{sec:analytical_solution} below. Explicit expressions for the optimal control are also given in Section~\ref{sec:analytical_solution}.}
\end{proof}

\subsection{Phase Portrait analysis\label{subsec:phase_potrait}}

\begin{figure}[t]
\centering
\includegraphics[width=7.5cm]{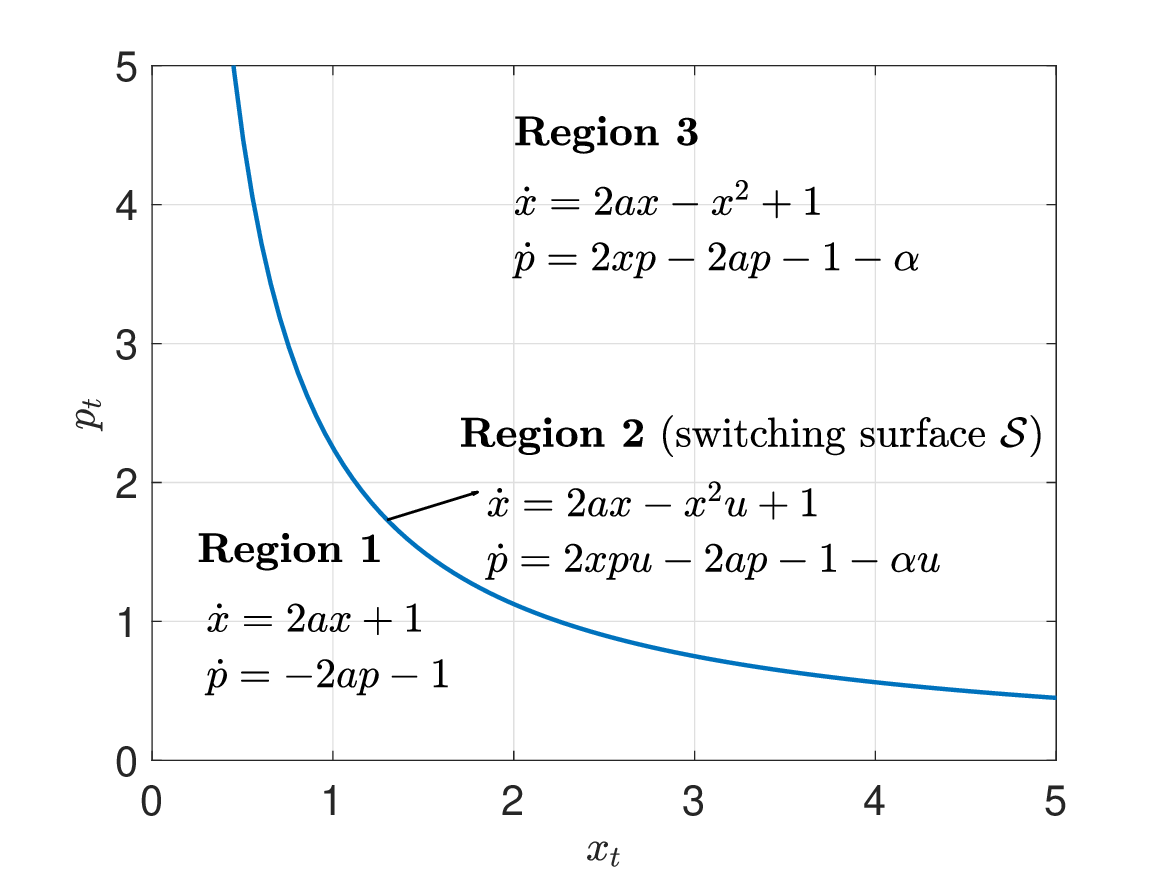}
\caption{Illustration of different regions in the phase space.}
\label{fig:regions}
\end{figure}
To analyze the canonical equations \eqref{eqn:canonical}, consider the vector field defined by the RHS of \eqref{eqn:canonical} and \eqref{eqn:u_star_argmin}. Because of the classification in \eqref{eqn:u_star_argmin}, the vector field is discontinuous at the switching surface $\mathcal{S}$ characterized by $p_tx_t=\alpha$. We divide the domain $\{(x,p)|x\geq0\}$ into Regions 1, 2 and 3 as shown in Fig.~\ref{fig:regions}.\footnote{Fig.~\ref{fig:regions} only shows the region $\{(x,p)|x\geq0, p\geq 0\}$ as it turns out that the region $\{(x,p)|x\geq0, p< 0\}$ plays no role in the following derivation.}
{\color{black} Denote the vector field in Region 1, 2 and 3 (see Fig. \ref{fig:cases}) as $f_1, f_2$ and $f_3$, respectively. From \eqref{eqn:canonical} and \eqref{eqn:u_star_argmin}, we have
\begin{align}
    f_1 :& \begin{cases} \dot{x}_t=2ax_t+1 \\
    \dot{p}_t=-2ap_t-1 \end{cases} \label{eqn:region1_xp_f1} \\
    f_2 :& \begin{cases}\dot{x}_t=2ax_t-x_t^2u_t+1 \\
    \dot{p}_t=2x_tp_tu_t-2ap_t-1-\alpha u_t \end{cases}  \label{eqn:region2_xp_f2} \\
    f_3 :& \begin{cases} \dot{x}_t=2ax_t-x_t^2+1\\
    \dot{p}_t=2x_tp_t-2ap_t-1-\alpha. \end{cases} \label{eqn:region3_xp_f3}
\end{align}
\subsubsection{Local solutions in Regions 1 and 3} 
Since \eqref{eqn:region1_xp_f1} is a linear differential equation, its general solution is given by
\begin{align}
    & x_t=\frac{k_1e^{2at}-1}{2a},\quad p_t=\frac{k_2e^{-2at}-1}{2a}
    \label{eqn:region1_xp}
\end{align}
where $k_1$ and $k_2$ are constants. Since \eqref{eqn:region3_xp_f3} belongs to a class of scalar Riccati differential equations, an analytical solution exists and is given as follows:
\begin{subequations}
\begin{align}
      &x_t=a+c-\frac{2c}{k_3e^{2 ct}+1},\label{eqn:region3_x}\\
      &    p_{t}= k_4\frac{\left(k_3e^{2 ct}+1\right)^2}{e^{2 ct}}+ \frac{(1+\alpha)\left(k_3e^{2 ct}+1\right)}{2ck_3e^{2 ct}}.\label{eqn:region3_p}
\end{align}
\label{eqn:region3_xp}
\end{subequations}
where $k_3$, $k_4$ are constants and $c=\sqrt{a^2+1}$.
\subsubsection{Stationary points} 
The nature of the phase portrait (namely the location of the stationary points) changes depending on the value of $\alpha$. Observing that $0<(a+\sqrt{a^2+1})^2<1/4a^2$ for all $a<0$, the following three cases can occur:
\begin{itemize}
    \item Case A: $1/4a^2<\alpha$. In this case, the phase portrait has a unique stationary point in Region 1, located at 
    \begin{align}
    E=(x_e,p_e)=(-1/2a,-1/2a). \label{eqn:equilibrium_region_1}
    \end{align}
    It is not possible for $f_2$ to have a stationary point in Region 2 no matter what value of $u_t\in[0,1]$ is chosen. A stationary point does not exist in Region 3 either.
    \item Case B: $(a+\sqrt{a^2+1})^2\leq\alpha\leq 1/4a^2$. In this case, no stationary point exists in Region 1 or in Region 3. However, the point 
    \begin{align}
    E=(x_e,p_e)=(\sqrt{\alpha},\sqrt{\alpha})\label{eqn:equilibrium_region_2}
    \end{align}
    in Region 2 is a stationary point if $u_t$ is set to be
    \begin{equation}
    \label{eq:u_star}
    u^\star=2a/\sqrt{\alpha}+1/\alpha.
    \end{equation}
    Under the Case B assumption that $(a+\sqrt{a^2+1})^2\leq\alpha\leq 1/4a^2$, the value of $u^\star$ from \eqref{eq:u_star} satisfies $0\leq u^\star \leq 1$. In Region 2, no other point can be a stationary point. 
    \item Case C: $\alpha<(a+\sqrt{a^2+1})^2$. In this case, the phase portrait has a unique stationary point 
    \begin{align}
    (x_e,p_e)=\left(a+\sqrt{a^2+1},\frac{1+\alpha}{2\sqrt{a^2+1}}\right)\label{eqn:equilibrium_region_3}
    \end{align}
    in Region 3. No stationary point can exists in Regions 1 and 2.
\end{itemize}
The vector field in each case is depicted in Fig.~\ref{fig:cases}.

\begin{figure*}[ht]
 \captionsetup[subfigure]{justification=centering}
 \centering
 \begin{subfigure}{0.33\textwidth}
{\includegraphics[scale=0.32]{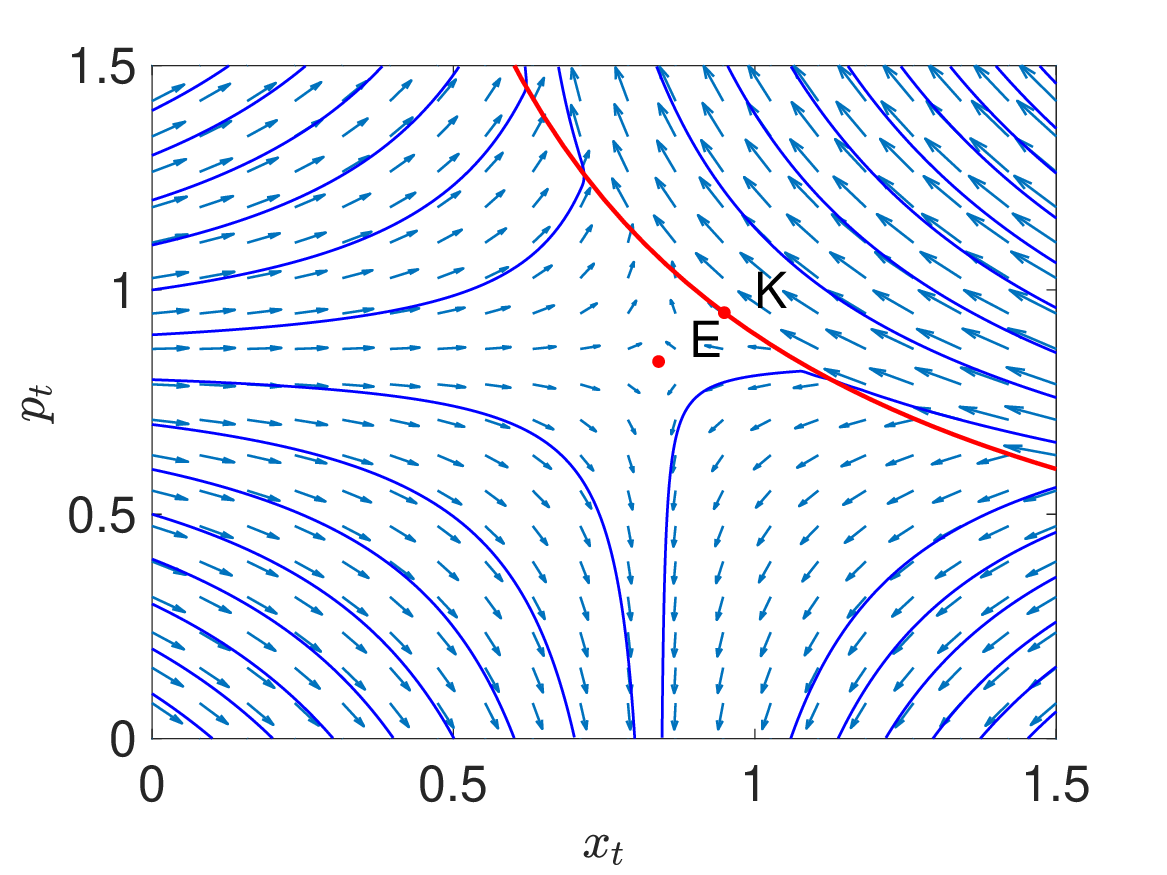}}
\caption{Case A}
\label{fig:casea}
 \end{subfigure}
\label{fig:approximation_error_x_position}
 \begin{subfigure}{0.33\textwidth}
{\includegraphics[scale=0.32]{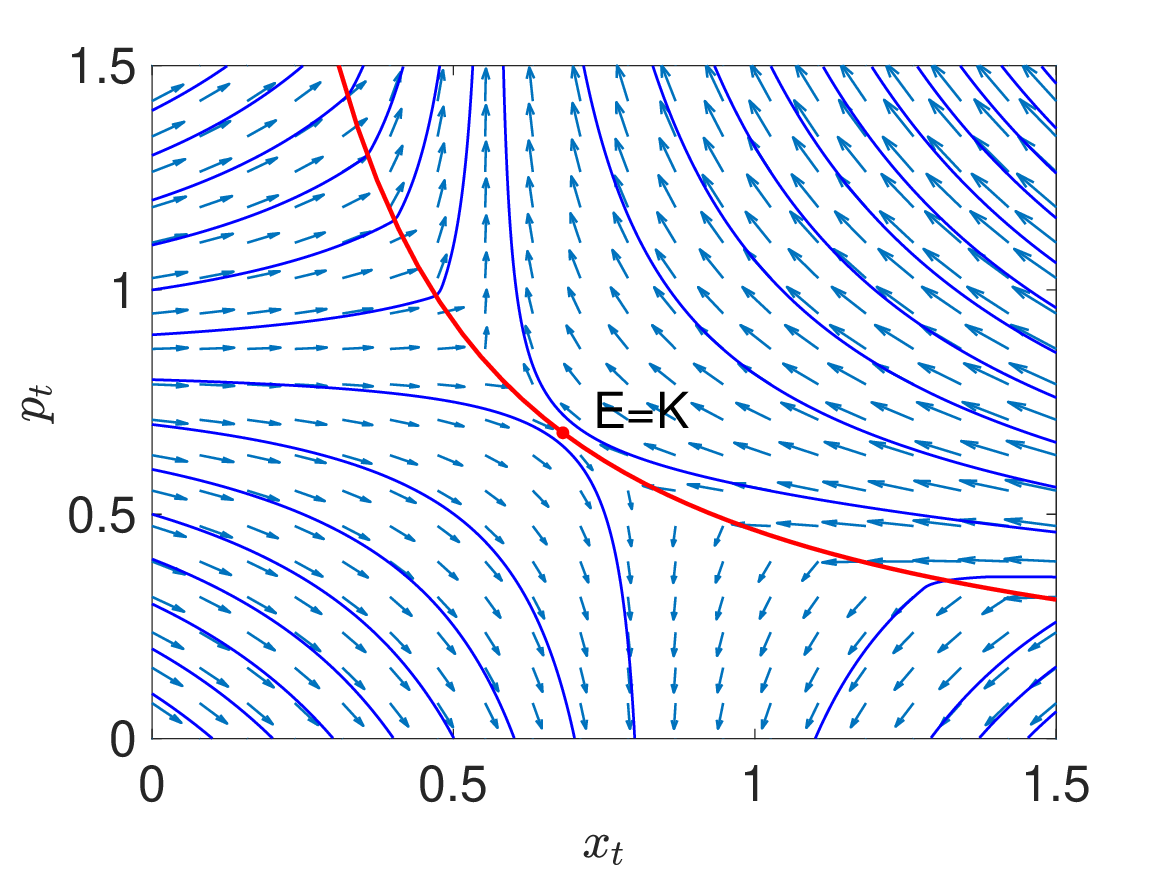}}
 \caption{Case B}
\label{fig:caseb}
 \end{subfigure}
\label{fig:approximation_error_v_velocity}
 \begin{subfigure}{0.32\textwidth}
{\includegraphics[scale=0.32]{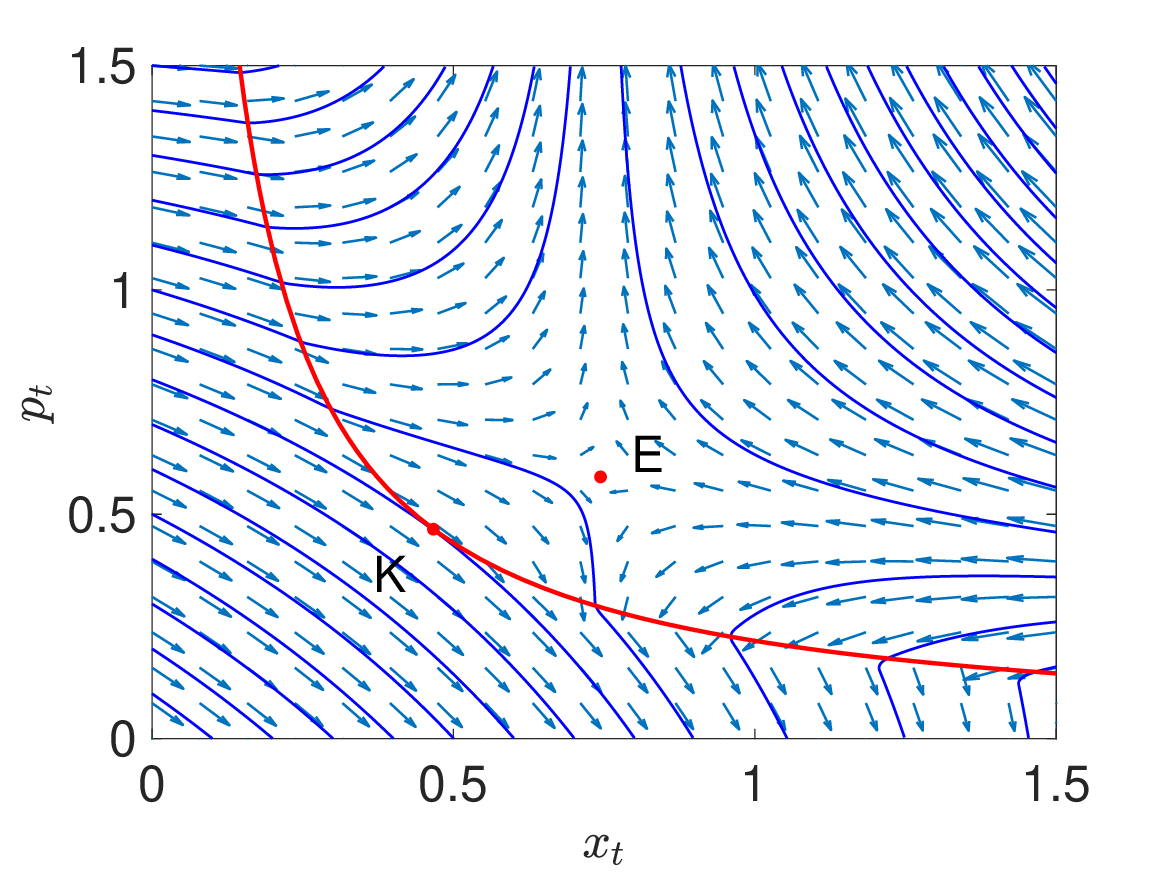}}
\caption{Case C}
\label{fig:casec}
\end{subfigure}
 \caption{Phase portraits for various cases.}
\label{fig:cases}
\end{figure*}

\subsubsection{Switching behavior}
To understand the switching behaviour, we analyze the directions of vector fields $f_1$, $f_2$ and $f_3$ with respect to $\mathcal{S}$ in the neighborhood of $\mathcal{S}$. Note that we can analyze the solution by checking the signs of Lie derivatives $L_{f_1}V$, $L_{f_2}V$ and $L_{f_3}V$ evaluated on surface $\mathcal{S}$. Note that the Lie derivatives along $f_1$, $f_2$ and $f_3$ are given by:
\begin{align}
    L_{f_1}V&=\frac{\partial V}{\partial x}\dot{x}+\frac{\partial V}{\partial p}\dot{p}=p\dot{x}+x\dot{p}\nonumber\\
    &=p(2ax+1)-x(2ap+1)=p-x\nonumber \\
    L_{f_2}V&=p\dot{x}+x\dot{p}=x^2pu+p-x-\alpha xu\nonumber \\      
L_{f_3}V&=p\dot{x}+x\dot{p}\nonumber\\
    &=p(2ax-x^2+1)+x(2xp-2ap-1-\alpha)\nonumber\\
        &=x^2p+p-x-\alpha x. \nonumber
\end{align}
Hence, when $px=\alpha$ (on $\mathcal{S}$), we have
\[
L_{f_1}V|_\mathcal{S}=L_{f_2}V|_\mathcal{S}=L_{f_3}V|_\mathcal{S}=\alpha/x-x.
\]
This would imply that $f_1$, $f_2$ and $f_3$ would point along the same direction everywhere on $\mathcal{S}$. More precisely, all these vector fields cross $\mathcal{S}$ downward when $x < \sqrt{\alpha}$, upward when $x < \sqrt{\alpha}$ and are tangential to $\mathcal{S}$ at point $K=(\sqrt{\alpha},\sqrt{\alpha})$. 
When $(a+\sqrt{a^2+1})^2\leq\alpha\leq 1/4a^2$ (Case B), the point $K$ becomes a stationary point and $u_t$ is defined as given in \eqref{eq:u_star}. This point is the same $E$ defined in \eqref{eqn:equilibrium_region_2}. 

An important consequence of the above analysis is that the phase portraits in Fig.~\ref{fig:cases} are devoid of the ``chattering" solutions in all three cases. Therefore, the solution concept of Caratheodory \cite{discontinuous_dynamical_system_cortes} is sufficient to describe the solutions crossing the switching surface $\mathcal{S}$.  
However, note that the uniqueness of the solution is lost in Case B. For example, consider the set of all trajectories $(x_t, p_t)$ that ``stay" on $E=K$ for an arbitrary duration as follows:
\begin{itemize}
    \item $(x_t, p_t)$ solves \eqref{eqn:region1_xp_f1} or \eqref{eqn:region3_xp_f3} for $t_0\leq t \leq t'$ with  $(x_{t'}, p_{t'})=(\sqrt{\alpha}, \sqrt{\alpha})$;
    \item $(x_t, p_t)=(\sqrt{\alpha}, \sqrt{\alpha})$ for $t'\leq t \leq t''$;
    \item $(x_t, p_t)$ solves \eqref{eqn:region1_xp_f1} or \eqref{eqn:region3_xp_f3} for $t'' \leq t_1$ with  $(\bar{x}_{t''}, \bar{p}_{t''})=(\sqrt{\alpha}, \sqrt{\alpha})$.
\end{itemize}
It is easy to verify that regardless of the choice of $t''(\geq t')$, all these trajectories are Caratheodory solutions to the canonical equations. 
}

\subsection{Analytical solution \label{sec:analytical_solution}}
{\color{black}
With the phase portraits shown in Fig.~\ref{fig:cases} in mind, we solve the boundary value problem \eqref{eqn:canonical} and \eqref{eqn:u_star_argmin} with the initial state constraint $x_{t_0}=x_0 (\geq 0)$ and the terminal costate constraint $p_{t_1}=0$. In the following, the solution to this boundary value problem is simply referred to as the optimal solution. We now consider Cases A, B and C sequentially.

\subsubsection{Case A}
In this scenario, the initial point (denoted by $(x_{t_0}^\star,p_{t_0}^\star)$) is either in the blue or yellow regions illustrated in Fig.~\ref{fig:case_a_color}. The boundaries of these two regions are characterized by the horizontal separatrix converging to point $E$ and the switching surface $\mathcal{S}$.
\subsubsection*{\underline{Subcase A-1} ($(x_{t_0}^\star,p_{t_0}^\star)$ is in the yellow region)}
The particular solution satisfying \eqref{eqn:region1_xp_f1} and the boundary constraints $x_{t_0}=x_0$ and $p_{t_1}=0$ are given as follows:
\begin{subequations}
\label{eqn:xp_case1}
    \begin{align}
   & \bar{x}_t=\frac{(2ax_0+1)e^{2a(t-t_0)}-1}{2a}\label{eqn:x_case1}\\
   & \bar{p}_t=\frac{e^{-2a(t-t_1)}-1}{2a}.\label{eqn:p_case1}
\end{align}
\end{subequations}
If $(\bar{x}_{t_0}, \bar{p}_{t_0})$ computed using \eqref{eqn:xp_case1} lies in Region 1 (i.e., $\bar{x}_{t_0}\bar{p}_{t_0} \leq \alpha$), then Subcase A-1 applies.  
There is no switching as the optimal solution lies in Region 1.
\begin{figure*}[ht]
 \captionsetup[subfigure]{justification=centering}
 \centering
 \begin{subfigure}{0.33\textwidth}
{\includegraphics[width=6.2cm]{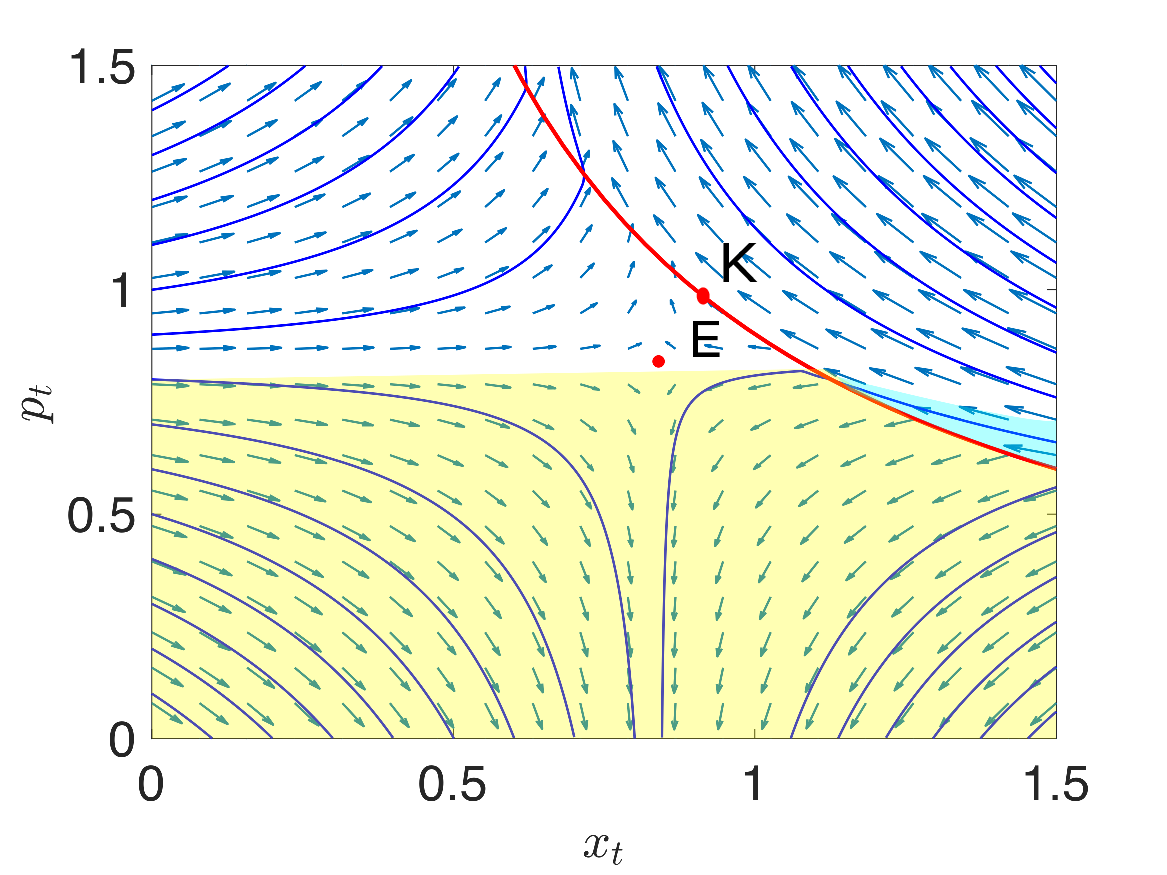}}
\caption{Case A.}
\label{fig:case_a_color}
 \end{subfigure}
 \begin{subfigure}{0.33\textwidth}
 \includegraphics[width=6.2cm]{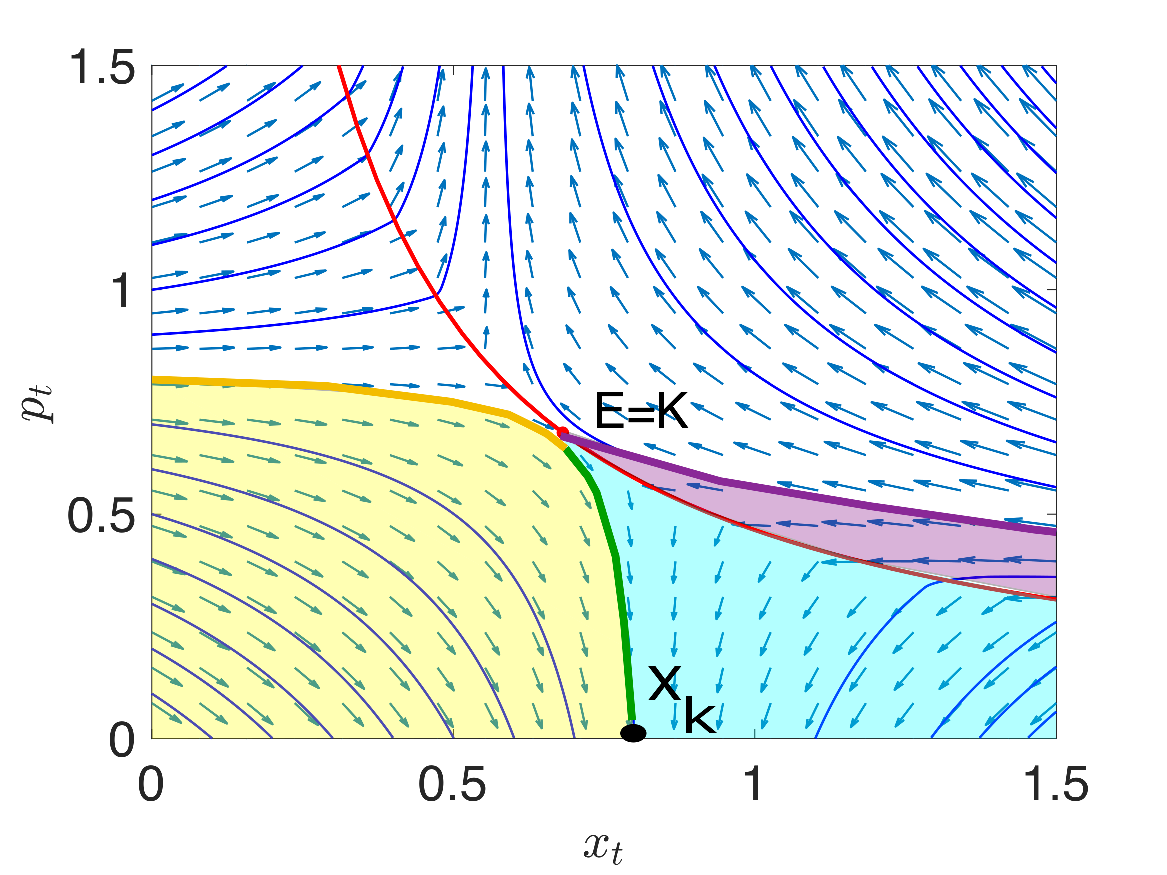}
 \caption{Case B.}
\label{fig:case_b_color}
 \end{subfigure}
 \begin{subfigure}{0.32\textwidth}
 \includegraphics[width=6.2cm]{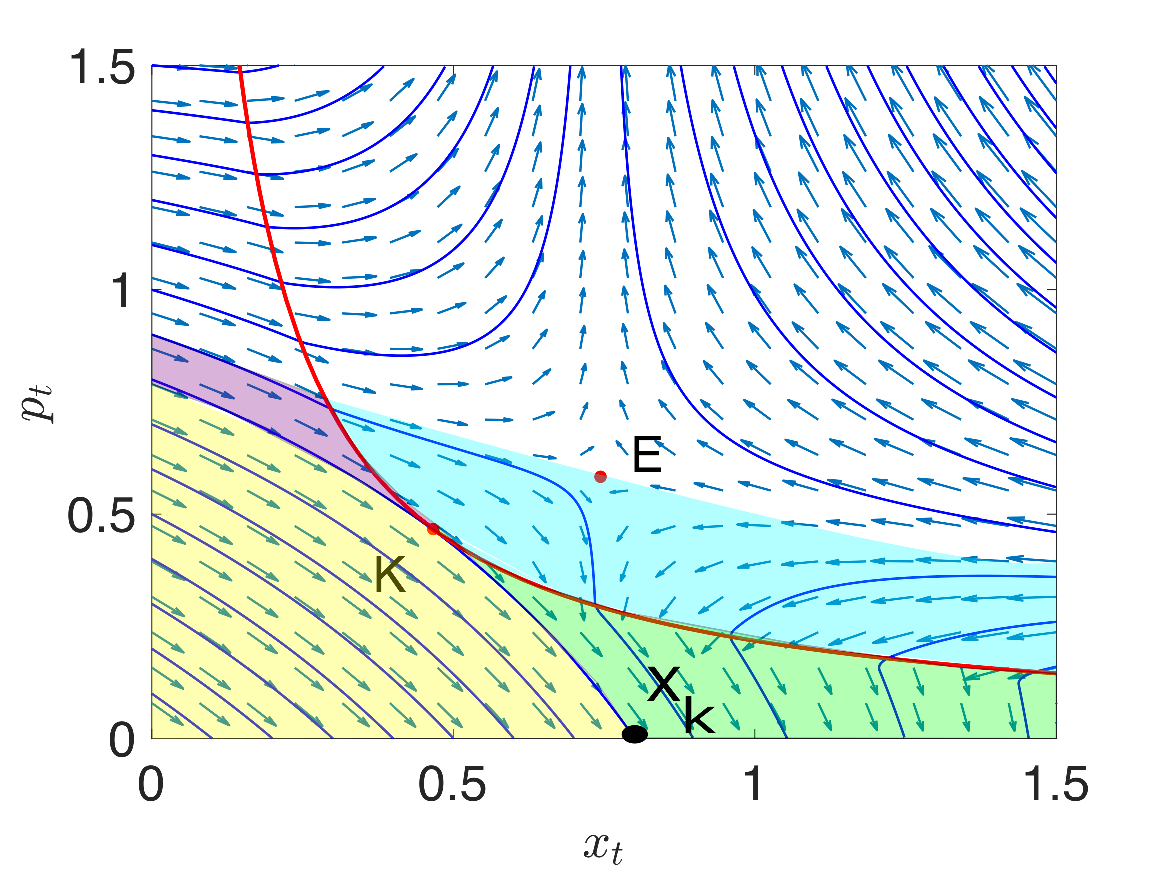}
\caption{Case C.}
\label{fig:case_c_color}
\end{subfigure}
 \caption{The different subcases for every particular case are depicted by colored regions .}
\label{fig:cases_color}
\end{figure*}
\subsubsection*{\underline{Subcase A-2} ($(x_{t_0}^\star,p_{t_0}^\star)$ is in the blue region)} \textcolor{blue}{}In this scenario, there exists a switching time $t'\in (t_0, t_1)$. To compute $t'$, observe that a particular solution \eqref{eqn:region3_x} satisfying the initial state constraint $x_{t_0}=x_0$ is given by
\begin{equation}
      \hat{x}_t=a+c-\frac{2c}{k_3e^{2 ct}+1} \label{eqn:x_case2}
\end{equation}
where $c=\sqrt{a^2+1}$ and
$k_3=\frac{c-a+x_0}{c+a-x_0}e^{-2ct_0}$.
On the contrary, the particular solution \eqref{eqn:region1_xp} satisfying $p_{t_1}=0$ is given by \eqref{eqn:p_case1}.
Therefore, at time $t'$, the state $\hat{x}_{t'}$ and costate $\bar{p}_{t'}$ must satisfy the governing switching surface equation i.e. $\hat{x}_{t'} \bar{p}_{t'}=\alpha$, or
\begin{align}
\left(a+c-\frac{2c}{k_3e^{2 ct'}+1}\right)\left(\frac{e^{-2a(t'-t_1)}-1}{2a}\right)=\alpha.
\label{eqn:caseA_switching_time}
\end{align}
Thereafter, we can compute $t'$ by solving  \eqref{eqn:caseA_switching_time}.


\subsubsection{Case B}
The initial point $(x_{t_0}^\star,p_{t_0}^\star)$ lies in either of the three colored regions shown in Fig.~\ref{fig:case_b_color}.
Let $x_K$ be the $x$-coordinate at which a particular solution \eqref{eqn:region1_xp} to the vector field $f_1$ that crosses $K=E=(\sqrt{\alpha},\sqrt{\alpha})$ at a particular time $t'' (< t_1)$ reaches at the terminal time $t_1$. It is trivial to show that $x_K=2a\alpha + 2\sqrt{\alpha}$ and the time $t''$ of crossing $K=E$ can also be computed from the following equation
\begin{equation}
\label{eq:case_b_t''}
t_1 -t'' = \frac{1}{2a}\ln(2a\sqrt{\alpha}+1).
\end{equation}
\subsubsection*{\underline{Subcase B-1} ($(x_{t_0}^\star,p_{t_0}^\star)$ is in the yellow region or on the green curve)}
Notice that the solution $(\bar{x}_t, \bar{p}_t)$ to $f_1$ with boundary conditions $x_{t_0}=x_0$ and $p_{t_1}=0$ is still given by \eqref{eqn:xp_case1}.
If the condition $\bar{x}_{t_1}\leq x_K$ is satisfied, then the optimal solution lies entirely in the yellow region and is given by \eqref{eqn:xp_case1}. In addition, switching does not occur in this case. 
\subsubsection*{\underline{Subcase B-2} ($(x_{t_0}^\star,p_{t_0}^\star)$ is in the blue region)}
This scenario arises when $\bar{x}_{t_0}\bar{p}_{t_0}\leq \alpha$, $\bar{x}_{t_1}> x_K$ and $x_0>\sqrt{\alpha}$. The optimal solution lies entirely in the light blue region and is given by \eqref{eqn:xp_case1}. In this case too, switching does not occur.
\subsubsection*{\underline{Subcase B-3} ($(x_{t_0}^\star,p_{t_0}^\star)$ is on the orange curve)}
This scenario arises when \eqref{eqn:xp_case1} satisfies $\bar{x}_{t_0}\bar{p}_{t_0}\leq \alpha$, $\bar{x}_{t_1}> x_K$ and $x_0 \leq \sqrt{\alpha}$. In this case, the particular solution \eqref{eqn:xp_case1} for $t_0 \leq t \leq t_1$ does not lie in Region 1 and thus it is not a reasonable solution to the boundary problem of our interest. The optimal solution in this scenario is represented by green and orange curves as shown in Fig.~\ref{fig:case_b_color}. First, the optimal solution follows the orange curve from $x_{t_0}=x_0$ to $x_{t'}=\sqrt{\alpha}$, stays on $E=K$ for $t'\leq t \leq t''$, and then traces the green trajectory from $x_{t''}=\sqrt{\alpha}$ to $x_{t_1}=x_K$.
From \eqref{eqn:x_case1}, the time $t'$ can be computed from
\begin{equation}
\label{eq:case_b_t'}
\sqrt{\alpha}=\frac{(2ax_0+1)e^{2a(t'-t_0)}-1}{2a}.
\end{equation}
Two switches occur in the optimal control input -- a switch at time $t'$ from $u=0$ to $u=u^\star$, and a switch at time $t''$ from $u=u^\star$ to $u=0$.
\subsubsection*{\underline{Subcase B-4} ($(x_{t_0}^\star,p_{t_0}^\star)$ is in the light purple region)}
This scenario arises when \eqref{eqn:xp_case1} satisfies $\bar{x}_{t_1}>x_K$ and $\bar{x}_{t_0}\bar{p}_{t_0}>\alpha$, and equation \eqref{eqn:caseA_switching_time} has a solution $t'$ in $[t_0, t_1]$. In this case, a single switching occurs at $t'$, from the light purple region to the blue region.
\subsubsection*{\underline{Subcase B-5} ($(x_{t_0}^\star,p_{t_0}^\star)$ is on the purple curve)}
This scenario arises when when \eqref{eqn:xp_case1} satisfies $\bar{x}_{t_0}\bar{p}_{t_0}>\alpha$ and $\bar{x}_{t_1}>x_K$, and equation \eqref{eqn:caseA_switching_time} does not have a solution $t'$ in $[t_0, t_1]$.
In this case, the optimal solution traces the purple trajectory from $x_{t_0}=x_0$ to $x_{t'}=\sqrt{\alpha}$, stays on $E=K$ for $t$ (where $t'\leq t \leq t''$), and then traces the green trajectory from $x_{t''}=\sqrt{\alpha}$ to $x_{t_1}=x_K$.
From \eqref{eqn:region3_x}, the time $t'$ can be computed from
\begin{equation}
\label{eq:case_b5_t'}
\sqrt{\alpha}=a+c-\frac{2c}{k_3e^{2 ct'}+1}
\end{equation}
where $k_3=\frac{c-a+x_0}{c+a-x_0}e^{-2ct_0}$. The optimal control input switches twice:  a switch at time $t'$ from $u=1$ to $u=u^\star$, and a switch at $t''$ from $u=u^\star$ to $u=0$.

\subsubsection{Case C}
In this scenario, the initial state-costate pair $(x_{t_0}^\star,p_{t_0}^\star)$ of the optimal solution can belong to four different regions as shown by different colors in Fig.~\ref{fig:case_c_color}. The boundaries of the blue region are characterized by $\mathcal{S}$ and the separatrices converging to $E$.
\subsubsection*{\underline{Subcase C-1} ($(x_{t_0}^\star,p_{t_0}^\star)$ is in the yellow region)} Let $x_K=2a\alpha+2\sqrt{\alpha}$ be the $x$-coordinate shown in Fig.~\ref{fig:case_c_color}.
Consider once again the trajectory \eqref{eqn:xp_case1} solving $f_1$ with the boundary constraints $p_{t_1}=0$ and $x_{t_0}=x_0$. If we have $\bar{x}_{t_1}\leq x_K$, then the optimal solution lies entirely in the yellow region and hence in this case, no switching occurs.
\subsubsection*{\underline{Subcase C-2} ($(x_{t_0}^\star,p_{t_0}^\star)$ is in the green region)} If $\bar{x}_{t_0}\bar{p}_{t_0}\leq \alpha$, $\bar{x}_{t_1}> x_K$ and $x_0 > \sqrt{\alpha}$ , then the trajectory \eqref{eqn:xp_case1} is lies in the green region entirely and therefore switching does not occur.
\subsubsection*{\underline{Subcase C-3} ($(x_{t_0}^\star,p_{t_0}^\star)$ is in the blue region)}
In this scenario, a switching occurs only once. Let $t'\in (t_0, t_1)$ be the switching time. Then the optimal solution traces the trajectory of the form
\begin{subequations}
\begin{align}
      &\hat{x}_t=a+c-\frac{2c}{k_3e^{2 ct}+1}, \\
      & \hat{p}_{t}= k_4\frac{\left(k_3e^{2 ct}+1\right)^2}{e^{2 ct}}+ \frac{(1+\alpha)\left(k_3e^{2 ct}+1\right)}{2ck_3e^{2 ct}}
\end{align}
\label{eqn:x_hat}
\end{subequations}
for $t_0\leq t \leq t'$, and 
\begin{align}
   \bar{p}_t=\frac{e^{-2a(t-t_1)}-1}{2a} 
   \label{eqn:p_bar}
\end{align}
for $t'\leq t\leq t_1$. Hence, the Subcase C-3 occurs only if the following set of equations in terms of unknowns $p_{t_0}, k_3, k_4$ and $t'$ has a solution that satisfies $x_0 p_{t_0} > \alpha$ and $t_0 < t' < t_1$:
\begin{subequations}
\label{eq:c3_condition}
\begin{align}
    &x_{t_0}^\star=x_0=a+c-\frac{2c}{k_3e^{2ct_0}+1} \label{eq:c3_condition1} \\
    & p_{t_0}^\star=k_4\frac{\left(k_3e^{2 ct_0}+1\right)^2}{e^{2 ct_0}}+ \frac{(1+\alpha)\left(k_3e^{2 ct_0}+1\right)}{2ck_3e^{2 ct_0}} \label{eq:c3_condition2} \\
    &\left(a+c-\frac{2c}{k_3e^{2 ct'}+1}\right)\left(\frac{e^{-2a(t'-t_1)}-1}{2a}\right)=\alpha \label{eq:c3_condition3} \\
    &k_4\frac{\left(k_3e^{2 ct'}+1\right)^2}{e^{2 ct'}}+ \frac{(1+\alpha)\left(k_3e^{2 ct'}+1\right)}{2ck_3e^{2 ct'}} \nonumber \\
    &\hspace{28ex}=\frac{e^{-2a(t'-t_1)}-1}{2a} \label{eq:c3_condition4}
\end{align}
\end{subequations}
The condition \eqref{eq:c3_condition3} guarantees that $\hat{x}_{t'}\bar{p}_{t'}=\alpha$. Further, \eqref{eq:c3_condition4} ensures that $\hat{p}_{t'}=\bar{p}_{t'}$ (i.e., the transition from $\widehat{p}_t$ to $\bar{p}_t$ is continuous).
\subsubsection*{\underline{Subcase C-4} ($(x_{t_0}^\star,p_{t_0}^\star)$ is in the purple region)}
Switching occurs twice in this case. Let $t'$ and $t''$ be the first and the second switching times. The optimal solution follows the trajectory of the form
\[
\tilde{x}_t=\frac{k_1e^{2at}-1}{2a},\quad \tilde{p}_t=\frac{k_2e^{-2at}-1}{2a}
\]
for $t_0\leq t \leq t'$ and satisfies \eqref{eqn:x_hat}
for $t'\leq t \leq t''$. For $t''\leq t_1$, the $p$-coordinate of the optimal solution satisfies \eqref{eqn:p_bar}.
Hence, Subcase C-4 occurs only if the following set of equations in terms of unknowns $k_1, k_2, k_3, k_4, p_{t_0}, t'$ and $t''$ has a solution such that $t_0 \leq t' < t'' < t_1$ and $x_0p_{t_0}\leq \alpha$:
\begin{subequations}
\label{eq:c4_condition}
\begin{align}
    &x_{t_0}^\star=x_0=\frac{k_1e^{2at_0}-1}{2a} \label{eq:c4_condition1} \\
    &p_{t_0}^\star=\frac{k_2e^{-2at_0}-1}{2a} \label{eq:c4_condition2} \\ 
    &\left(\frac{k_1e^{2at'}-1}{2a}\right)\left(\frac{k_2e^{-2at'}-1}{2a}\right)=\alpha \label{eq:c4_condition3} \\
    &\frac{k_1e^{2at'}-1}{2a}=a+c-\frac{2c}{k_3e^{2 ct'}+1} \label{eq:c4_condition4} \\
    &\frac{k_2e^{-2at'}-1}{2a}=k_4\frac{\left(k_3e^{2 ct'}+1\right)^2}{e^{2 ct'}}+\nonumber\\
    &\quad\quad\quad\quad\quad\quad\quad\quad\quad\frac{(1+\alpha)\left(k_3e^{2 ct'}+1\right)}{2ck_3e^{2 ct'}} \label{eq:c4_condition5} \\
    &\left(a+c-\frac{2c}{k_3e^{2 ct''}+1}\right)\left(\frac{e^{-2a(t''-t_1)}-1}{2a}\right)=\alpha \label{eq:c4_condition6} \\
    &k_4\frac{\left(k_3e^{2 ct''}+1\right)^2}{e^{2 ct''}}+ \frac{(1+\alpha)\left(k_3e^{2 ct''}+1\right)}{2ck_3e^{2 ct''}} \nonumber \\
    &\hspace{28ex}=\frac{e^{-2a(t''-t_1)}-1}{2a} \label{eq:c4_condition7}
\end{align}
\end{subequations}
Conditions \eqref{eq:c4_condition3} and \eqref{eq:c4_condition6} ensure that $\tilde{x}_{t'}\tilde{p}_{t'}=\alpha$ and $\bar{x}_{t''}\hat{p}_{t''}=\alpha$  (i.e., switching occurs on $\mathcal{S}$). Conditions \eqref{eq:c4_condition4}, \eqref{eq:c4_condition5} and \eqref{eq:c4_condition7} guarantees that the trajectory is continuous at times when switching occurs.

}

\section{Conclusion and Future work\label{sec:conclusion}}
{
In this paper, we formulated a finite-horizon optimal sensor gain control problem for minimum-information estimation of continuous-time Gauss-Markov processes.
We established the existence of an optimal control based on Filippov's theorem and a necessary optimality condition based on Pontryagin's minimum principle. For scalar cases, we computed the optimal solution explicitly and showed that the optimal control is piecewise constant with switching at most twice.
Future work includes the computation of the optimal control for vector cases, analysis of the solution for infinite-horizon problems, and the applications of the obtained results to sensor data compression problems.
}

 \bibliography{main.bib}

\begin{thebibliography}{10}
\providecommand{\url}[1]{#1}
\csname url@samestyle\endcsname
\providecommand{\newblock}{\relax}
\providecommand{\bibinfo}[2]{#2}
\providecommand{\BIBentrySTDinterwordspacing}{\spaceskip=0pt\relax}
\providecommand{\BIBentryALTinterwordstretchfactor}{4}
\providecommand{\BIBentryALTinterwordspacing}{\spaceskip=\fontdimen2\font plus
\BIBentryALTinterwordstretchfactor\fontdimen3\font minus
  \fontdimen4\font\relax}
\providecommand{\BIBforeignlanguage}[2]{{%
\expandafter\ifx\csname l@#1\endcsname\relax
\typeout{** WARNING: IEEEtran.bst: No hyphenation pattern has been}%
\typeout{** loaded for the language `#1'. Using the pattern for}%
\typeout{** the default language instead.}%
\else
\language=\csname l@#1\endcsname
\fi
#2}}
\providecommand{\BIBdecl}{\relax}
\BIBdecl

\bibitem{duncan1970calculation}
T.~E. Duncan, ``On the calculation of mutual information,'' \emph{SIAM Journal
  on Applied Mathematics}, vol.~19, no.~1, pp. 215--220, 1970.

\bibitem{kadota1971mutual}
T.~Kadota, M.~Zakai, and J.~Ziv, ``Mutual information of the white
  \textsc{G}aussian channel with and without feedback,'' \emph{IEEE
  Transactions on Information theory}, vol.~17, no.~4, pp. 368--371, 1971.

\bibitem{guo2005mutual}
D.~Guo, S.~Shamai, and S.~Verd{\'u}, ``Mutual information and minimum
  mean-square error in \textsc{G}aussian channels,'' \emph{IEEE transactions on
  information theory}, vol.~51, no.~4, pp. 1261--1282, 2005.

\bibitem{palomar2005gradient}
D.~P. Palomar and S.~Verd{\'u}, ``Gradient of mutual information in linear
  vector \textsc{G}aussian channels,'' \emph{IEEE Transactions on Information
  Theory}, vol.~52, no.~1, pp. 141--154, 2005.

\bibitem{gallego2019event}
G.~Gallego, T.~Delbruck, G.~Orchard, C.~Bartolozzi, B.~Taba, A.~Censi,
  S.~Leutenegger, A.~Davison, J.~Conradt, K.~Daniilidis \emph{et~al.},
  ``Event-based vision: A survey,'' \emph{IEEE Transactions on Pattern Analysis
  and Machine Intelligence}, 2020.

\bibitem{tanaka2017optimal}
T.~Tanaka, M.~Skoglund, and V.~Ugrinovskii, ``Optimal sensor design and
  zero-delay source coding for continuous-time vector
  \textsc{G}auss-\textsc{M}arkov processes,'' \emph{2017 IEEE 56th Annual
  Conference on Decision and Control (CDC)}, pp. 3992--3997, 2017.

\bibitem{gorbunov1974prognostic}
A.~Gorbunov and M.~S. Pinsker, ``Prognostic epsilon entropy of a
  \textsc{G}aussian message and a \textsc{G}aussian source,'' \emph{Problemy
  Peredachi Informatsii}, vol.~10, no.~2, pp. 5--25, 1974.

\bibitem{tatikonda2004stochastic}
S.~Tatikonda, A.~Sahai, and S.~Mitter, ``Stochastic linear control over a
  communication channel,'' \emph{IEEE transactions on Automatic Control},
  vol.~49, no.~9, pp. 1549--1561, 2004.

\bibitem{derpich2012improved}
M.~S. Derpich and J.~Ostergaard, ``Improved upper bounds to the causal
  quadratic rate-distortion function for \textsc{G}aussian stationary
  sources,'' \emph{IEEE Transactions on Information Theory}, vol.~58, no.~5,
  pp. 3131--3152, 2012.

\bibitem{charalambous2013nonanticipative}
C.~D. Charalambous, P.~A. Stavrou, and N.~U. Ahmed, ``Nonanticipative rate
  distortion function and relations to filtering theory,'' \emph{IEEE
  Transactions on Automatic Control}, vol.~59, no.~4, pp. 937--952, 2013.

\bibitem{tanaka2016semidefinite}
T.~Tanaka, K.-K.~K. Kim, P.~A. Parrilo, and S.~K. Mitter, ``Semidefinite
  programming approach to \textsc{G}aussian sequential rate-distortion
  trade-offs,'' \emph{IEEE Transactions on Automatic Control}, vol.~62, no.~4,
  pp. 1896--1910, 2016.

\bibitem{kostina2019rate}
V.~Kostina and B.~Hassibi, ``Rate-cost tradeoffs in control,'' \emph{IEEE
  Transactions on Automatic Control}, vol.~64, no.~11, pp. 4525--4540, 2019.

\bibitem{tanaka2017lqg}
T.~Tanaka, P.~M. Esfahani, and S.~K. Mitter, ``\textsc{LQG} control with
  minimum directed information: Semidefinite programming approach,'' \emph{IEEE
  Transactions on Automatic Control}, vol.~63, no.~1, pp. 37--52, 2017.

\bibitem{bansal1989simultaneous}
R.~Bansal and T.~Ba{\c{s}}ar, ``Simultaneous design of measurement and control
  strategies for stochastic systems with feedback,'' \emph{Automatica},
  vol.~25, no.~5, pp. 679--694, 1989.

\bibitem{filippov1962certain}
A.~Filippov, ``On certain questions in the theory of optimal control,''
  \emph{Journal of the Society for Industrial and Applied Mathematics, Series
  A: Control}, vol.~1, no.~1, pp. 76--84, 1962.

\bibitem{optimal_control_daniel_liberzon}
D.~Liberzon, \emph{Calculus of Variations and Optimal Control Theory: A Concise
  Introduction}.\hskip 1em plus 0.5em minus 0.4em\relax USA: Princeton
  University Press, 2011.

\bibitem{athans2013optimal_book}
M.~Athans and P.~L. Falb, \emph{Optimal control: an introduction to the theory
  and its applications}.\hskip 1em plus 0.5em minus 0.4em\relax Courier
  Corporation, 2013.

\bibitem{angeli2003monotone}
D.~Angeli and E.~D. Sontag, ``Monotone control systems,'' \emph{IEEE
  Transactions on automatic control}, vol.~48, no.~10, pp. 1684--1698, 2003.

\bibitem{discontinuous_dynamical_system_cortes}
J.~Cortes, ``Discontinuous dynamical systems,'' \emph{IEEE Control Systems
  Magazine}, vol.~28, no.~3, pp. 36--73, 2008.

\end{thebibliography}
\end{document}